\documentclass[a4paper,10pt]{article}
\usepackage{latexsym}
\usepackage{theorem}
\usepackage{color,graphicx}
\usepackage{amssymb}
\usepackage{amsmath}
\usepackage{txfonts}
\usepackage{enumitem}
\usepackage{fullpage}

\usepackage{theorem}
\newtheorem{theorem}{Theorem}[section]
\newtheorem{lemma}{Lemma}[section]
\newtheorem{corollary}{Corollary}[section]
\newtheorem{proposition}{Proposition}[section]

\usepackage{algorithmic}
\usepackage{algorithm,caption}

\floatname{algorithm}{Procedure}

\newcommand{\view}[1][]{\if!#1!\mathcal{V}\else\mathcal{V}^{#1}\fi}
\newcommand{\cQ}{\mathcal{Q}}
\newcommand{\partition}[1]{\Pi_{#1}}
\newcommand{\viewPath}[3][]{\if!#1! P(#2,#3)\else P(#2,#1,#3)\fi}  
\newcommand{\ourColor}{\alpha}
\newcommand{\st}{\hspace{0.1cm}\bigl|\bigr.\hspace{0.1cm}}
\newcommand{\card}[1]{|#1|}
\newcommand{\levelInView}[2]{\textup{lev}_{#2}(#1)}   
\newcommand{\diam}{D}
\newcommand{\commonTime}[1]{\xi_{#1}}
\newcommand{\commonTimeMax}{\Xi}
\newcommand{\modulo}{\,\textup{mod}\,}

\newcommand{\nats}{\mathbb{N}}

\newcommand{\algorithmMain}{\textup{\texttt{Solve-LE-and-TOP}}}
\newcommand{\algorithmQuotient}{\textup{\texttt{ComputeQuotientGraph}}}
\newcommand{\algorithmView}{\textup{\texttt{ComputeView}}}
\newcommand{\algorithmTest}{\textup{\texttt{TestRepetition}}}

\newcommand{\returnTrue}{\texttt{true}}
\newcommand{\returnFalse}{\texttt{false}}
\newcommand{\returnFailure}{\texttt{{unsolvable}}}
\newcommand{\communicate}{\texttt{communicate}}

\newcommand{\problemLE}{\textup{LE}}
\newcommand{\problemTOP}{\textup{TOP}}

\newcommand{\cV}{{\cal V}}
\newcommand{\qed}{\hfill $\square$ \smallbreak}
\newenvironment{proof}[1][Proof]
{\par\noindent{\bf #1:} }{\hspace*{\fill}\nolinebreak{$\Box$}\bigskip\par}

\begin{document}

\title{Topology Recognition and Leader Election in Colored Networks}

\author{Dariusz Dereniowski\thanks{Faculty of Electronics, Telecommunications and Informatics, Gda{\'n}sk University of Technology, Narutowicza 11/12, 80-233 Gda\'{n}sk, Poland. Email: deren@eti.pg.gda.pl. Partially supported by National Science Centre grant DEC-2015/17/B/ST6/01887.}
\and
Andrzej Pelc\thanks{D\'epartement d'informatique, Universit\'e du Qu\'ebec en Outaouais, Gatineau, Qu\'ebec J8X 3X7, Canada. Email: pelc@uqo.ca. Partly supported by the NSERC discovery grant 8136 -- 2013 and by the Research Chair in Distributed Computing at the Universit\'e du Qu\'ebec en Outaouais.}
}

\date{}
\maketitle


\begin{abstract}
Topology recognition and leader election are fundamental tasks in distributed computing in  networks. The first of them requires each node to find a labeled isomorphic copy of the network, while the result of the second one consists in a single node adopting the label 1 (leader), with all
other nodes adopting the label 0 and learning a path to the leader. We consider both these problems in networks whose nodes 
are equipped with not necessarily distinct labels called {\em colors}, and ports at each node of degree $d$ are arbitrarily numbered $0,1,\dots, d-1$. Colored networks are generalizations both of labeled networks, in which nodes have distinct labels, and of anonymous 
networks, in which nodes do not have labels (all nodes have the same color). 

In colored networks, topology recognition and leader election are not always feasible. Hence we study two more general problems.
Consider a colored network and an input $I$ given to its nodes.
The aim of the problem $\problemTOP$, for this colored network and for $I$, is to solve topology recognition in this network, if this is possible under input $I$, and to have all nodes answer ``unsolvable'' otherwise.
Likewise, the aim of the problem $\problemLE$ is to solve leader election in this network, if this is possible under input $I$, and to have all nodes answer ``unsolvable'' otherwise. 

We show that nodes of a network can solve problems $\problemTOP$ and $\problemLE$, if they are given, as input $I$,
an upper bound $k$ on the number of nodes of a given color, called the {\em size} of this color.
On the other hand we show that, if the nodes are given an input that does not bound the size of any color, then the answer to $\problemTOP$ and $\problemLE$ must be ``unsolvable'', even for the class of rings.

Under the assumption that nodes are given an upper bound $k$ on the size of a given color, we study the time of solving problems $\problemTOP$ and $\problemLE$ in the $\cal{LOCAL}$ model in which, during each round, each node can exchange arbitrary messages with all its neighbors and perform arbitrary local computations.
We give an algorithm to solve each of these problems in arbitrary networks in time $O(kD+D\log(n/D))$, where $D$ is the diameter
of the network and $n$ is its size. We also show that this time is optimal, by exhibiting
classes of networks in which every algorithm solving problems $\problemTOP$ or $\problemLE$ must use time
$\Omega(kD+D\log(n/D))$.
\end{abstract}


\section{Introduction} \label{sec:intro}

\subsection{The model and the problem} \label{subsec:model}

Topology recognition and leader election are fundamental tasks in distributed computing in  networks. The goal of topology recognition is for each node of the network to acquire a faithful map of it (an isomorphic
copy of the underlying network with all nodes having distinct identifiers), with the position of the node marked in the map. If nodes can solve this problem, any other 
distributed task, such as leader election \cite{HS,P}, minimum weight spanning tree construction \cite{A}, 
renaming \cite{ABDKPR}, etc. can be performed by them using
only local computations. Thus topology recognition converts all distributed problems to centralized ones, 
in the sense that nodes can solve any distributed problem simulating a central monitor.

Leader election, first stated in \cite{LL}, is likewise of fundamental importance.
Each node of the network has a Boolean variable initialized to 0 and, after the election, exactly one node,
called the {\em leader}, should change this value to 1. All other nodes should know which one becomes the leader by discovering a path to it.
Notice that the above two problems are equivalent: having a map of the network with distinct node labels, nodes can elect the node with the smallest
label as the leader, and conversely, knowing a leader, nodes can construct a map of the network using the leader as a stationary token, cf. \cite{CDK}.

It should be noted that formulations of the leader election problem vary across the literature (cf. \cite{Ly}). In a weak formulation, every node should only know if it is the leader or not. (In \cite{MP}, this task was called selection, by contrast to election). In a strong formulation, every node should moreover get to know who is the leader. We adopt the latter formulation of the leader election problem. When nodes have distinct identities, knowing who is the leader means outputting its identity. In our scenario distinct identities need not exist, hence knowing who is the leader means that every node outputs a path to the leader, coded as a sequence of ports. This formulation of leader election was used, e.g., in \cite{GMP} for anonymous networks.

A network is modeled as a simple undirected connected graph. As commonly done in the literature, cf., e.g. \cite{YK3} {or $KT_0$ model in \cite{AGPV}}, we assume that
ports at a node of degree $d$ have arbitrary fixed labelings $0,\dots,$ $d-1$.
We do not assume any coherence between port labelings at various nodes.
As for nodes, we assume that they are equipped with not necessarily distinct labels called {\em colors}.  In applications, colors may be types of the devices interconnected by the network, 
such as workstations, servers, laptops, or mobile phones.
Networks with colored nodes are generalizations both of labeled networks in which nodes have distinct labels, and of anonymous 
networks, in which nodes do not have labels (all nodes have the same color). 
Nodes communicate by exchanging arbitrary messages along links. A node sending a message
through a given port appends the port number to the message, and a node receiving a message
through a port is aware of the port number by which the message is received.

If nodes have distinct identities, both topology recognition and leader election are easily accomplished in any network. By contrast, it is well known that, in the absence of distinct node labels, these tasks are often impossible, if no additional information about the network is provided to nodes. In fact, even the less demanding task of reconstructing an unlabeled isomorphic copy of the
network is sometimes impossible. 
For example, in an anonymous ring whose each edge has port numbers 0 and 1 at its
endpoints, not only topology recognition and leader election cannot be achieved but even 
the size of the ring cannot be learned by nodes. Providing the size of the network as input is not a remedy
either: the authors of \cite{YK3} give examples of two (anonymous) non-isomorphic graphs of size 6 whose
nodes cannot decide in which of these two graphs they are.

Due to these impossibilities, we consider two problems more general than topology recognition
and leader election, respectively.
Consider a colored network and an input $I$ given to its nodes.
The aim of the problem $\problemTOP$, for this colored network and for $I$, is to solve topology recognition in this network, if this is possible under input $I$, and to have all nodes answer ``unsolvable'' otherwise.
Likewise, the aim of the problem $\problemLE$ is to solve leader election in this network, if this is possible under input $I$, and to have all nodes answer ``unsolvable'' otherwise. 

Our goal is to find out what type of input has to be given to the nodes of a colored network in order to enable them to solve problems $\problemTOP$ and $\problemLE$, and what is the minimal time in which they can solve these problems, if this input is provided.
To investigate time, we use the extensively studied $\cal{LOCAL}$ model \cite{Pe}.
In this model, communication proceeds in synchronous rounds and all nodes start simultaneously.
In each round each node can exchange arbitrary messages with all its neighbors and perform arbitrary local computations. 
The time of completing a task is the number of rounds it takes.

\subsection{Our results} \label{subsec:our_results}

We first show that nodes of a network can solve problems $\problemTOP$ and $\problemLE$, if they are given
an upper bound $k$ on the number of nodes of a given color, called the {\em size} of this color.
This means that, if such an upper bound is known to all nodes (even if they do not know any upper bound on the total number of nodes
or on the number of colors), then they can correctly decide if leader election and topology recognition
are feasible in the given network, and if so, they can perform these tasks. 
On the other hand, if the nodes are given an input that does not bound the size of any color, then the answer to $\problemTOP$ and $\problemLE$ must be ``unsolvable'', even for the class of rings.

Hence, providing all nodes with an upper bound on the size of some color is the weakest assumption under which problems $\problemTOP$ and $\problemLE$ can be meaningfully solved. 
This justifies the use of this assumption in our algorithms.

Next, assuming that all nodes have an upper bound $k$ on the size of a given color,
we study the time of solving problems $\problemTOP$ and $\problemLE$ in the $\cal{LOCAL}$ model. We give an 
algorithm to solve each of these problems in arbitrary networks in time $O(kD+D\log(n/D))$, where $D$ is the diameter
of the network and $n$ is its size. We also show that this time is optimal, by exhibiting
classes of networks in which every algorithm solving problems $\problemTOP$ or $\problemLE$ must use time
$\Omega(kD+D\log(n/D))$.

\subsection{Related work} \label{subsec:related_work}
Early studies of leader election in networks mostly concerned the scenario where all nodes have distinct labels.
This task was first studied for rings.
A synchronous algorithm, based on comparisons of labels, and using
$O(n \log n)$ messages was given in \cite{HS}. It was proved in \cite{FL} that
this complexity is optimal for comparison-based algorithms. On the other hand, the authors showed
an algorithm using a linear number of messages but requiring very large running time.
An asynchronous algorithm using $O(n \log n)$ messages was given, e.g., in \cite{P} and
the optimality of this message complexity was shown in \cite{B}. Deterministic leader election in radio networks has been studied, e.g., 
in \cite{JKZ,KP,NO} and randomized leader election, e.g., in \cite{Wil}. In \cite{HKMMJ} the leader election problem is
approached in a model based on mobile agents for networks with labeled nodes.

Many authors \cite{An,ASW,AtSn,BV,DKMP,Kr,KKV,Saka,YK,YK3} studied leader election
and other computational problems
in anonymous networks. In particular, \cite{BSVCGS,YK3} characterize message passing networks in which
leader election can be achieved when nodes are anonymous. 
The authors assume that nodes know an upper bound on the size of the network. In \cite{YK2} the authors study
the problem of leader election in general networks, under the assumption that labels are
not unique. They characterize networks in which this can be done and give an algorithm
which performs election when it is feasible. They assume that the number of nodes of the
network is known to all nodes. In
 \cite{FKKLS}  the authors
study feasibility and message complexity of sorting and leader election in rings with
nonunique labels, while in \cite{DoPe} the authors provide algorithms for the
generalized leader election problem in rings with arbitrary labels,
unknown (and arbitrary) size of the ring, and for both
synchronous and asynchronous communication. 
Characterizations of feasible instances for leader election and naming problems have been provided in~\cite{C,CMM,CM}.
Memory needed for leader election in unlabeled networks has been studied in \cite{FP}. 
In \cite{FP1}, the authors investigated the time of leader election in anonymous networks
by characterizing this time in terms of the size and diameter of the network, and of an additional
parameter, called level of symmetry, which measures how deeply nodes have to inspect the network to notice differences in their views of it.
In \cite{DP1}, the authors studied feasibility of leader election among anonymous agents that
navigate in a network in an asynchronous way.

Feasibility of topology recognition for anonymous networks with adversarial port labelings was studied in~\cite{YK3},
under the assumption that nodes know an upper bound on the size of the network.
The problem of efficiency of map construction by a mobile agent, equipped with a token and exploring an anonymous network,  has
been studied in \cite{CDK}. In \cite{DP}, the authors investigated the minimum size of advice
that has to be given to a mobile agent, in order to enable it to reconstruct  the topology of an anonymous network or to construct its spanning tree.
In \cite{FPR},
tradeoffs between time of topology recognition and the size of advice given to nodes were studied
in the $\cal{LOCAL}$ communication model.

\section{Preliminaries} \label{sec:feasibility}

In this section we introduce some basic terminology and provide preliminary  results known from the literature. Let $G$ be a simple connected undirected network with the set of nodes $V$, and let $c$ be a positive integer.
Consider any surjective function $f: V \longrightarrow \{1,\dots , c\}$.
The couple $(G,f)$ is called a {\em colored network}, the function $f$ is called a
{\em coloring} of this network, and $f(v)$ is called the {\em color} of node $v$.

We will use the following notion from \cite{YK3}. Let $G$ be a network and $v$ a node of $G$.  We first define, for any $l \geq 0$,  the {\em truncated view}
$\cV^l(v)$ at depth $l$, by induction on $l$. $\cV^0(v)$ is a tree consisting of a single node $x_0$. 
If $\cV^l(u)$ is defined for any node $u$ in the network, then $\cV^{l+1}(v)$ is the port-labeled tree
rooted at $x_0$ and defined as follows.
For every node $v_i$, $i=1,\dots ,k$, adjacent to $v$, 
there is a child $x_i$ of $x_0$ in $\cV^{l+1}(v)$ such that the port number at $v$ corresponding to edge $\{v,v_i\} $ is the same as the port number 
at $x_0$ corresponding to edge $\{x_0,x_i\}$,
and the port number at $v_i$ corresponding to edge $\{v,v_i\} $ is the same as the port number at $x_i$ corresponding to edge $\{x_0,x_i\}$.  We say that the node $x_i$ {\em represents} node $v_i$.
Now node $x_i$, for $i=1,\dots ,k$ becomes 
the root of the truncated view $\cV^l(v_i)$.   
 The {\em view} of $v$ is the infinite rooted tree $\cV(v)$ with labeled ports, such that $\cV^l(v)$ is its truncation to depth $l$, for each $l\geq 0$.

We will also use a notion similar to that of the view but corresponding to colored networks (cf. \cite{Norris}). Consider a colored network $(G,f)$. Let $v$ be any node of $G$. 
Let $f^*: \cV(v) \longrightarrow \{1,\dots , c\} $ be the function defined as follows:
$f^*(x)=f(v)$, where  $x$ is a node of $\cV(v)$ representing node $v$.
The couple $(\cV(v), f^*)$ is called the {\em colored view} of node $v$.
Thus, the colored view of a node additionally marks colors of nodes represented in it.
The couple $(\cV^{l}(v), f_l)$, where $f_l$ is the truncation of $f^*$ to $\cV^l(v)$, is called a 
{\em truncated  colored view} of node $v$.

For the (truncated) views of a node $v$ we will often omit the node $v$ in the notation,
thus writing $\cV$ instead of $\cV (v)$ and $\cV^l$ instead of $\cV^l(v)$, if the node $v$, called
the {\em root} of the view, is clear from the context. The same convention applies to 
colored (truncated) views. The \emph{level} $l$ of a view $\cV$, denoted $\levelInView{\cV}{l}$, is the set of all its nodes at distance $l$ (in $\cV$) from the root of the view. For nodes $u$ and $v$ in the truncated view $\cV ^l$, we denote by $P(\cV ^l,u,v)$ the unique path in $\cV ^l$ from $u$ to $v$,
defined as a sequence of nodes in this truncated view. We denote by $P(\cV ^l,v)$ the unique path from the root of the view to $v$. For such a path $P$, we denote by $|P|$ the length of this path,
defined as the number of edges in it. For a truncated view $\cV ^l$ and a node $z$ in this view,
we denote by $\cV ^l[z]$ the subtree of  $\cV ^l$ rooted at $z$.

The following proposition was proved in \cite{H}.
\begin{proposition}\label{trunc}
For a $n$-node network of diameter $D$,
$\cV(u)=\cV(v)$, if and only if $\cV^{h}(u)=\cV^{h}(v)$,
for some $h \in \Theta(D \log (n/D))$.
\qed\end{proposition}

The following proposition that follows from \cite{DKP} shows that the truncation level $h$ from Proposition \ref{trunc}
is the smallest possible, up to constant factors.

\begin{proposition}\label{trunc-lower}
For any integers $D\leq n$, there exists a network $G$ of size $\Theta(n)$ and diameter $\Theta(D)$, with nodes $u$ and $v$, both with unique views,
such that $\cV(u)\neq \cV(v)$ but $\cV^{h'}(u)=\cV^{h'}(v)$,
for some $h' \in \Theta(D \log (n/D))$.
Moreover, there exists a network $G'$ having the same size and diameter as that of $G$, with a node $u'$, such that $\view[h'](u')=\view[h'](u)$.
 \qed\end{proposition}
 
Propositions \ref{trunc} and \ref{trunc-lower} remain valid, when views are replaced by colored views.

Define the following equivalence relations on the set of nodes of a colored network $(G,f)$.
{Let} $u\sim v$ if and only if $(\cV(u),f^*)=(\cV(v), f^*)$, and {let also} $u\sim_t v$ if and only if $(\cV^t(u), f_t)=(\cV^t(v), f_t)$.
Let $\Pi$ be the partition of all nodes into equivalence classes of $\sim$, and $\Pi_t$ the corresponding partition for  $\sim_t$.
It follows from \cite{YK3} that all equivalence classes in $\Pi$ are of equal size $\sigma$. In view of Proposition \ref{trunc}
this is also the case for $\Pi_t$, for some $t\in \Theta(D \log (n/D))$. On the other hand, for smaller $t$, equivalence classes in $\Pi_t$ 
may be of different sizes.  
Every equivalence class in  $\Pi_t$ is a union of some equivalence classes in  $\Pi_{t'}$, for $t<t'$.  
The following result was proved in \cite{Norris}. It says that if the sequence of partitions $\Pi_t$ stops changing at some point, it will never change again. 

\begin{proposition}\label{stop}
If $\Pi_t=\Pi_{t+1}$, then $\Pi_t=\Pi$.
\qed\end{proposition}

The following proposition, easily proved by induction on $T$, implies that if $\view[T](u)=\view[T](v)$ for some nodes $u$ and $v$, and all nodes are given the same information about the network, then any algorithm solving $\problemTOP$ or $\problemLE$ in time at most $T$ must give the same output, when executed by $u$ and by $v$.
\begin{proposition} \label{prop:algoView}
Let $u$ be a node in a colored network $G$ and $u'$ a node in a colored network $G'$.
Suppose that initially all nodes of $G$ and $G'$ have the same input.
Let $T$ be a positive integer and assume that $\view[T](u)=\view[T](u')$.
For any $t\leq T$, let $M$ be the message received by $u$ through port $p$ in round $t$.
Then, message $M$ is received by $u'$ through port $p$ in round $t$.
\qed
\end{proposition}

Next, we define the notion of a {\em colored quotient graph}, which is a generalization of the
notion of quotient graph introduced in  \cite{YK3}. Given a colored network $(G,f)$, its colored quotient
graph $(Q,\overline{f})$ is defined as follows.  Nodes of $Q$ are equivalence classes of the above defined relation $\sim$. If $a$ and $b$ are two such classes, there is an edge joining $a$ and $b$
in $Q$, with port number $p$ at $a$ and $q$ at $b$, if and only if there is an edge joining nodes $u$ and
$v$ in $G$, with port number $p$ at $u$ and $q$ at $v$, where $u$ belongs to the class $a$ and $v$ belongs to the class $b$. (Hence, unlike $G$, the graph $Q$ can have self-loops and multiple
edges.) The function $\overline{f}$ is defined on all nodes of $Q$ by the formula
$\overline{f}(a)=f(u)$, where $u$ belongs to the class $a$. 

Finally, we give the formal definitions of the main problems $\problemTOP$ and $\problemLE$ considered in this paper. Both these problems are to be solved in an unknown colored graph.
For each of these problems, all nodes are given some common input $I$.
In order to solve the problem $\problemLE$, every node has to output a sequence of port numbers leading from this node to a single node, called the leader, if this task is possible to perform using input $I$;
otherwise, all nodes must output the answer ``unsolvable''. In order to solve the problem $\problemTOP$, every node $v$ has to output an isomorphic copy $C_v$ of the underlying graph, with all nodes labeled by distinct identifiers, and the node $v$ correctly marked in $C_v$,  if this task is possible to perform using input $I$;
otherwise, all nodes must output the answer ``unsolvable''.

\section{The algorithm and its analysis}  \label{sec:algorithm}

In this section we describe and analyze an algorithm for solving the problems $\problemTOP$ and $\problemLE$, assuming that some bound $k$ on the size of one of the colors is known.
This color will be denoted by $\ourColor$.
Recall that, by definition, there exists at least one node with color $\ourColor$ in the network.

The algorithm aims at computing the colored quotient graph.
Once this is achieved, the solutions to $\problemTOP$ and $\problemLE$ will follow easily.
The task of computing the colored quotient graph is divided into three procedures called $\algorithmTest$, $\algorithmView$ and $\algorithmQuotient$.
Procedure $\algorithmView$ computes the truncated colored view of the executing node up to a certain depth $l$.
This depth $l$ depends on the answers returned by several calls to $\algorithmTest$ in procedure $\algorithmView$.
Finally, procedure $\algorithmQuotient$ uses the view and, by further extending it to an appropriate depth, obtains the colored quotient graph.
We now give the detailed description of the above procedures.
They are all formulated for an executing node $w$.
We write $\view$ and $\view[l]$ instead of $\view(w)$ and $\view[l](w)$, respectively.

\medskip
Procedure $\algorithmTest$ uses the following notion of distance between colors and nodes.
Given a truncated colored view $(\view[l],f)$ and a node $v$ belonging to it, the \emph{distance from $v$ to color $\ourColor$} is the length of the shortest path in $\view[l][v]$ that connects $v$ with some node with color $\ourColor$.
Given a view $\view[l]$ and a node $v$ belonging to it, we say that $v$ has a \emph{copy in} $\view[i]$, $i\leq l$, if there exists $v'$ in $\view[i]$ such that $v'$ and $v$ represent the same node of the network.
Additionally, if $|\viewPath{\view[l]}{v'}|<|\viewPath{\view[l]}{v}|$, then we say that $v$ has a \emph{high copy} in $\view[l]$.

Before giving the pseudocode of procedure $\algorithmTest$, we describe its high-level idea.
Given a truncated colored view $(\view[l],f)$ and a node $v$ in it as an input, the goal of this procedure is to return $\returnTrue$ if $v$ has a high copy in $\view[l]$.
This is done by exploiting the only tool available to the algorithm: counting nodes in $\view[l]$ of color $\ourColor$.
Namely, for each node $u$ in $\viewPath{\view[l]}{v}$, the procedure computes the distance from $u$ to color $\ourColor$ in $\view[l][u]$ and then it takes the maximum over all such distances, denoted by $d'$.
If no node with the color $\ourColor$ is observed in $\view[l][u]$, then $d_u$ is set to $\infty$ which results in $d'=\infty$ and the procedure returns $\returnFalse$.
As proven below (cf. Lemma~\ref{lem:test:bound}) checking whether the distance from the root to $v$ in $\view[l]$ is at least $2(k+1)(d'+1)$ is sufficient for procedure $\algorithmTest$ to return $\returnTrue$.
Note that there exist networks such that $\card{\viewPath{\view[l]}{v}}$ is smaller than $2(k+1)(d'+1)$ but $v$ still has a high copy in $\view[l]$.
However, in such cases the procedure $\algorithmTest$ is unable to certify that and, thanks to later calls to $\algorithmTest$ in our algorithm, several descendants of $v$ will be recognized to have high copies.

\begin{algorithm} \caption{$\algorithmTest( \view[l],f,v )$}
\begin{algorithmic}
\REQUIRE Truncated colored view $(\view[l],f)$, $l\geq 1$, a node $v$ in $\view[l]$.
\ENSURE $\returnTrue$ or $\returnFalse$.

\FORALL{$u$ in $\viewPath{\view[l]}{v}$}
   \STATE $d_u \leftarrow$ distance from $u$ to color $\ourColor$ in $(\view[l][u],f)$;
   \STATE (Possibly $d_u=\infty$ if $\ourColor$ does not appear in $(\view[l][u],f)$.)
\ENDFOR
\STATE $d' \leftarrow \max\{d_u \st u\textup{ in }\viewPath{\view[l]}{v}\}$
\IF{ $\card{\viewPath{\view[l]}{v}} \geq 2(k+1)(d'+1)$ }
   \RETURN $\returnTrue$
\ELSE
   \RETURN $\returnFalse$
\ENDIF
\end{algorithmic}
\end{algorithm}

\begin{lemma} \label{lem:test:bound}
Let $k$ be an upper bound on the size of the color $\ourColor$ in the network {and let $v$ be a node in the network}.
Let $d'$ be the maximum distance from a node $u$ to color $\ourColor$ in $\view[l](u)$ taken over all nodes $u$ in $\viewPath{\view[l]}{v}$.
If $\card{\viewPath{\view[l]}{v}} \geq 2(k+1)(d'+1)$, then $v$ has a high copy in $\view[l]$.
\end{lemma}

\begin{proof}
Find an ancestor $u$ of $v$ in $\view[l]$ such that the length of the path from the root to $u$ is exactly $2(k+1)(d'+1)$.
Let $u_0,\ldots,u_{k+1}$ be such nodes in the path $\viewPath{\view[l]}{u}$ that $u_i$ is at distance $2i(d'+1)$ from the root.
Note that $u_{k+1}=u$.
For each $i\in\{1,\ldots,k+1\}$, let $z_i$ be a node with color $\ourColor$ at distance at most $d'$ from $u_i$ in $\view[l][u_i]$.
Such a node $z_i$ exists by definition of $d'$ and by the fact that $d'\neq\infty$, {for any} $i\in\{1,\ldots,k+1\}$.
The latter follows from $\card{\viewPath{\view[l]}{v}} \geq 2(k+1)(d'+1)$.

Since the size of color $\ourColor$ is at most $k$, there exist $i,i'\in\{1,\ldots,k+1\}$, $i<i'$, such that $z_i$ and $z_{i'}$ represent the same node of the network.
Consider the sequence of ports that is a concatenation of the sequences of ports of the following paths: $\viewPath{\view[l]}{z_i}$ and $\viewPath[z_{i'}]{\view[l]}{u_{k+1}}$.
This sequence of ports gives a valid path from the root to some node $u'$ in the colored view $\view$, since $z_i$ and $z_{i'}$ represent the same node of the network.
We bound the length $l'$ of $\viewPath{\view}{u'}$ as follows.
Note that the first part of $\viewPath{\view}{u'}$ of length $|\viewPath{\view}{z_i}|$ equals $\viewPath{\view}{z_i}$.
Thus, the length of this first part is at most
\[|\viewPath{\view[l]}{u_i}|+|\viewPath[u_i]{\view[l]}{z_i}| \leq 2i(d'+1) + d'.\]
The length of the second part, i.e., the length of $\viewPath[z_{i'}]{\view[l]}{u_{k+1}}$ is {at most}
\[|\viewPath[z_{i'}]{\view[l]}{u_{i'}}|+|\viewPath[u_{i'}]{\view[l]}{u_{k+1}}| \leq d' + 2(k+1-i')(d'+1).\]
Therefore, since $i<i'$, we obtain
\[|\viewPath{\view}{u'}|\leq 2d'+2i(d'+1)+2(k+1-i')(d'+1)<2(k+1)(d'+1).\]
Thus, $u'$ belongs to $\view[2(k+1)(d'+1)-1]$, or in other words, the node $u$ that belongs to level $2(k+1)(d'+1)$ of $\view[l]$ has a high copy in $\view[2(k+1)(d'+1)-1]$.
Since either $u=v$ or $u$ is an ancestor of $v$, we obtain that $v$ has a high copy in $\view[l]$, as required.
\end{proof}

\begin{corollary} \label{cor:test:bound}
Consider a truncated colored view $(\view[l],f)$ and a node $v$ in this view.
If procedure $\algorithmTest$ returns $\returnTrue$ for input $\view[l],f$ and $v$, then $v$ has a high copy in $\view[l]$.
\qed
\end{corollary}

We have proved that if procedure $\algorithmTest$ returns $\returnTrue$, then this guarantees that the input node $v$ has a high copy.
However, for the correctness of our final algorithm we need to ensure that each infinite simple path in $\view$ originating from the root contains a node $v$ that has a high copy, and that this fact will be detected by procedure $\algorithmTest$.
Moreover, in order to bound the time of our final algorithm, we need to estimate the distance from such $v$ to the root, which is done in the next lemma.
\begin{lemma} \label{lem:test:finite}
Let $(\view[l],f)$ be a truncated colored view and let $v$ be a node in it.
If $l=2(k+1)(\diam+1)+\diam$ and $v$ belongs to level $2(k+1)(\diam+1)$ of $\view[l]$, then procedure $\algorithmTest$ executed for $\view[l],f$ and $v$ returns $\returnTrue$.
\end{lemma}
\begin{proof}
For each node $u$ of the path $\viewPath{\view[l]}{v}$, the node of the network represented by it is at distance at most $\diam$ (in the network) from some node with color $\ourColor$.
This implies that $d_u\leq\diam$ for each such node $u$.
Thus, $d'\leq\diam$.
We obtain that
\[\card{\viewPath{\view[l]}{v}} = 2(k+1)(\diam+1) \geq 2(k+1)(d'+1),\]
which implies that procedure $\algorithmTest$ returns true.
\end{proof}

\medskip
Before describing the next procedure, we introduce some more notation.
A truncated colored view $(\view[l],f)$ \emph{covers the network} if, for each node $x$ of the network, there exists a node $u$ of $\view[l]$ such that $u$ represents $x$.
We define procedure $\communicate$ which sends the currently acquired truncated colored view $(\view[t],f)$ to all neighbors and receives the messages containing currently acquired colored views of the same depth $t$ from all neighbors.
Note that after all nodes have performed $\communicate$ $t$ times, each node can compute its truncated colored view $\view[t]$.

We now describe procedure $\algorithmView$.
Again, we start with an informal description.
In each iteration, the `while' loop increments the depth of the view currently stored at the executing node.
This is done by communicating with each neighbor.
The crucial part is to decide when to stop.
At some point, procedure $\algorithmView$ detects that the currently possessed truncated colored view $(\view[l],f)$ covers the network and this view is then returned.
The above is achieved (see Lemma~\ref{lem:covered} below for a proof) by maintaining a set $M$, that is initially empty, consisting of nodes having high copies.
Procedure $\algorithmView$ stops when each leaf of $\view[l]$ is in $M$, which guarantees that $\view[l]$ covers the network, as required.
\begin{algorithm} \caption{$\algorithmView$}
\begin{algorithmic}
\REQUIRE None.
\ENSURE Truncated colored view $(\view[l],f)$.

\STATE $l \leftarrow 1$
\STATE $M \leftarrow \emptyset$
\WHILE{there exists a leaf in $\view[l]$ that is not in $M$}
   \FORALL{$v$ in $\view[l]$}
      \IF{$\algorithmTest( \view[l],v )$ returns $\returnTrue$}
         \STATE Add to $M$ the node $v$ and all its descendants in $\view[l]$.
      \ENDIF
   \ENDFOR
   \STATE $\communicate$~$\quad$ \COMMENT{This extends $(\view[l],f)$ to $(\view[l+1],f)$.}
   \STATE $l \leftarrow l+1$
\ENDWHILE
\RETURN $(\view[l],f)$
\end{algorithmic}
\end{algorithm}

\begin{lemma} \label{lem:covered}
Procedure $\algorithmView$ returns the truncated colored view $(\view[l],f)$ of the executing node, such that $\view[l]$ covers the network and $l\leq 2(k+1)(\diam+1)+\diam$.
\end{lemma}
\begin{proof}
Lemma~\ref{lem:test:finite} implies that there exists $l\leq 2(k+1)(\diam+1)+\diam$ such that all leaves of $\view[l]$ are in $M$ and hence the number of iterations of the `while' loop of procedure $\algorithmTest$ is at most $l$.
The latter relies on the fact that if a node having a high copy is detected by procedure $\algorithmTest$, then procedure $\algorithmView$ adds to $M$ this node with all its descendants in the current truncated view.

Consider the view $\view$ of the executing node, and let $u$ be any node in $\view$.
We argue that $u$ has a copy in $\view[l]$.
Suppose for a contradiction that this is not the case and select $u$ to be a node that does not have a copy in $\view[l]$ and is closest to the root in $\view$.
Let $l'$ be the level of $u$ in $\view$.
(Clearly $l'>l$.)
Since $\levelInView{\view[l]}{l}\subseteq M$, there exists an ancestor $v$ of $u$ such that $v\in\levelInView{\view[l]}{i}$, $i\leq l$, and procedure $\algorithmTest$ returns $\returnTrue$ when executed for $\view[l],f$ and $v$.
By Corollary~\ref{cor:test:bound}, $v$ has a copy $v'$ in $\view[i-1]$.
Thus, $u$ has a copy $u'$ in $\view[l'-1]$.
But then, by the minimality of $l'$, $u'$ has a copy in $\view[l]$, which is also a copy of $u$.
This is a contradiction because, by definition, $u$ and $u'$ represent the same node of the network.
\end{proof}

\medskip
Procedure $\algorithmQuotient$ computes the colored quotient graph $(\cQ,g)$ of the network, provided that a colored view $(\view[l],f)$ that covers the network is given as an input.
This is done by finding the minimum index $i$ such that $\partition{i-1}=\partition{i}$.
Note that this requires that each node learn its colored view till depth $l+i$, by exchanging messages with its neighbors.
\begin{algorithm} \caption{$\algorithmQuotient( \view[l],f )$}
\label{alg:ComputeQuotientGraph}
\begin{algorithmic}
\REQUIRE Truncated colored view $(\view[l],f)$, $l\geq 1$.
\ENSURE  The colored quotient graph $(\cQ,g)$.

\STATE $\partition{-1} \leftarrow \emptyset$
\STATE $\partition{0} \leftarrow \{ (\view[0](v),f) \st v\in\view[l] \}$
\STATE $i\leftarrow 0$
\WHILE{$\partition{i}\neq\partition{i-1}$}
   \STATE $i \leftarrow i+1$
   \STATE $\communicate$
   \STATE Compute $\partition{i}$
\ENDWHILE
\STATE Compute the labeled quotient graph $(\cQ,g)$ using $\partition{i}$ and $\view[l+i]$
\RETURN $(\cQ,g)$
\end{algorithmic}
\end{algorithm}

We prove the following.
\begin{lemma} \label{lem:computeQG}
Let $(\view[l],f)$ be the truncated colored view  computed by procedure $\algorithmView$.\\
Procedure $\algorithmQuotient$ called for $(\view[l],f)$ correctly computes the colored quotient graph of the network.
\end{lemma}
\begin{proof}
By Lemma~\ref{lem:covered}, $\view[l]$ covers the network.
Thus, $\partition{0}$ contains all nodes of the network.
By Proposition~\ref{stop}, the partition $\partition{i}$ obtained in the last iteration of the `while' loop of procedure $\algorithmQuotient$ equals $\partition{}$.
Hence, this partition $\partition{i}$ is the set of all nodes of the colored quotient graph.
This implies that the quotient graph can be computed on the basis of $\partition{i}$ and $\view[l]$ as follows.
The color of a node of the quotient graph is set to the color of its elements in the truncated colored view.
The edges and port numbers are added as in the definition of the quotient graph.
\end{proof}

In the formulation of our main algorithm we will use the following integer that each node can compute once it has the colored quotient graph $(\cQ,g)$.
For any node $u$ of the quotient graph, let $\commonTime{u}$ be the sum of running times of procedures $\algorithmView$ and $\algorithmQuotient(\view[l],f)$, where $(\view[l],f)$ is computed by $\algorithmView$.
Let $\commonTimeMax$ be the maximum of $\commonTime{u}$ over all nodes $u$ of the quotient graph.

We are ready to state our main algorithm for solving the problems $\problemLE$ and $\problemTOP$.
We formulate it as a single procedure, since all steps leading to the computation of the colored quotient graph are identical in both cases.

\floatname{algorithm}{Algorithm}
\begin{algorithm} \caption{$\algorithmMain(k,\ourColor)$}
\begin{algorithmic}
\REQUIRE An upper bound $k$ on the size of color $\ourColor$.
\ENSURE  For $\problemLE$ --- a sequence of port numbers leading from the executing node to the leader, or $\returnFailure$, if leader election is impossible.
         For $\problemTOP$ --- the topology of the network, or $\returnFailure$, if topology recognition is impossible.

\STATE $(\view[l],f) \leftarrow \algorithmView$
\STATE $(\cQ,g) \leftarrow \algorithmQuotient( \view[l],f )$
\STATE Perform $\communicate$ $\commonTimeMax$ times.
\IF{the number of nodes with color $\ourColor$ in $(\cQ,g)$ is at most $\lfloor k/2\rfloor$ and $\cQ$ is not a tree}
   \RETURN $\returnFailure$
\ELSE
   \STATE For $\problemLE$ --- return a sequence of port numbers of the path $\viewPath{\view[l]}{v}$, where $v$ (which is the leader) corresponds to the node of $(\cQ,g)$ whose colored view is lexicographically smallest (if the executing node is the leader, then the path is empty).
   \STATE For $\problemTOP$ --- return $(\cQ,g)$.
\ENDIF
\end{algorithmic}
\end{algorithm}

\begin{theorem} \label{thm:main}
If a bound $k$ on the size of a given color $\ourColor$ is provided as an input, then Algorithm $\algorithmMain$ correctly solves problems $\problemLE$ and $\problemTOP$ in time $O(k\diam+\diam\log(n/\diam))$, where $n$ is the size of the network and $D$ is its diameter.
\end{theorem}
\begin{proof}
First note that Algorithm $\algorithmMain$ correctly computes the colored quotient graph.
Indeed, by Lemma~\ref{lem:computeQG}, Algorithm $\algorithmMain$ obtains the colored quotient graph $(\cQ,g)$ as a result of the call to procedure $\algorithmQuotient$.
This is ensured by the fact that each node performs at least $\commonTimeMax$ calls to $\communicate$ and therefore every node $u$ can compute $\view[\commonTime{u}]$, which is enough to compute the colored quotient graph, by definition of $\commonTime{u}$.

Once the colored quotient graph is computed by all nodes of the network, the correctness of Algorithm $\algorithmMain$ essentially follows from \cite{YK3}.
For completeness we include the short argument.

Assume that the colored quotient graph has at most $\lfloor k/2\rfloor$ nodes with color $\ourColor$ and it is not a tree (i.e., a graph {that has} cycles, multiple edges or self-loops).
Then, topology recognition cannot be solved since there exist two non-isomorphic networks of size $2\lfloor k/2\rfloor$ having $(\cQ,g)$ as a quotient graph.
Any potential topology recognition algorithm in these networks must have the same execution for each pair of nodes with the same colored views and thus such an algorithm must be incorrect.
Hence, Algorithm $\algorithmMain$ correctly returns $\returnFailure$ for the problem $\problemTOP$.
For the problem of leader election, take any network of size $2\lfloor k/2\rfloor$ having $(\cQ,g)$ as a quotient graph.
Two distinct nodes $u$ and $v$ in this network have the same colored view.
Thus, any potential leader election algorithm incorrectly elects at least two leaders in this network.
Hence, Algorithm $\algorithmMain$ correctly returns $\returnFailure$ for the problem $\problemLE$ as well.

Otherwise, i.e., if the colored quotient graph has more than $\lfloor k/2\rfloor$ nodes with color $\ourColor$ or it is a tree, then the network is isomorphic to $(\cQ,g)$ and hence Algorithm $\algorithmMain$ gives a correct solution to the problem $\problemTOP$.
As for leader election, each node has a unique colored view under this assumption.
Hence, the node with the lexicographically smallest colored view is unambiguously elected as the leader by each node.

It remains to bound the time of computation.
It is at most $l+i+\commonTimeMax$, where $l$ and $i$ are the numbers of iterations of the `while' loop of procedure $\algorithmView$ and $\algorithmQuotient$, respectively.
By Lemma~\ref{lem:test:finite} and the formulation of procedure $\algorithmView$, $l\in O(k\diam)$.
By Proposition~\ref{trunc}, $i\in O(\diam\log(n/\diam))$.
Thus, by definition, $\commonTimeMax\in O(k\diam+\diam\log(n/\diam))$, which completes the proof.
\end{proof}

\begin{corollary} \label{cor:main}
Let $(G,f)$ be a colored network and let $k$ be a bound on the size of color $\ourColor$.
If all nodes of $(G,f)$ have pairwise different colored views and the size of color $\ourColor$ is strictly greater than $\lfloor k/2\rfloor$, then topology recognition and leader election are possible in $G$.
\qed
\end{corollary}

\section{Negative results}

In this section we present our negative results. The first of them {is an impossibility result indicating} that if no upper bound on the size of any color is given to nodes, then problems $\problemTOP$ and $\problemLE$ must have answer ``unsolvable'', even if other restrictions on the possible sizes of colors are known. In order to express this result in full generality, we formalize such possible restrictions as a set $R\subseteq \nats^r$, where $\nats$ denotes the set of positive integers and $r$ is the number of colors.
We translate the property that no upper bound on the size of any color is known, to the statement
that for any point $(x_1,x_2,\dots , x_r)\in \nats^r $ there exists a point $(y_1,y_2,\dots , y_r)$  in the restriction set $R$, such that $y_i \geq x_i$, for all $i \leq r$.
Such a set $R$ will be called \emph{unbounded}.
(An example of an input defining an unbounded restriction set is: there are three colors and the sizes of all of them are prime integers.)
To make the impossibility result even stronger, we prove that it holds even for a very simple class of networks: on rings. 

\begin{proposition}
Let $(C,g)$ be any colored ring.
Consider an input $I$ defining an unbounded restriction set $R$.
Then, problems $\problemTOP$ and $\problemLE$ must have answer ``unsolvable''.
\end{proposition}
\begin{proof}
Take any algorithm, call it $A$, that correctly solves problem $\problemTOP$ or problem $\problemLE$ in any colored ring under input $I$.
Let $r$ be the number of different colors that appear in the network, and let $x_i$ be the size of color $i$, $i\in\{1,\ldots,r\}$.
Denote the nodes of $C$ by $v_0,\ldots,v_{n-1}$, where $n=x_1+\cdots+x_r$ and $v_i$ is adjacent to $v_{(i+1)\modulo n}$, for each $i\in\{0,\ldots,n-1\}$.
\begin{figure}[htb]
\begin{center}
\includegraphics[scale=0.9]{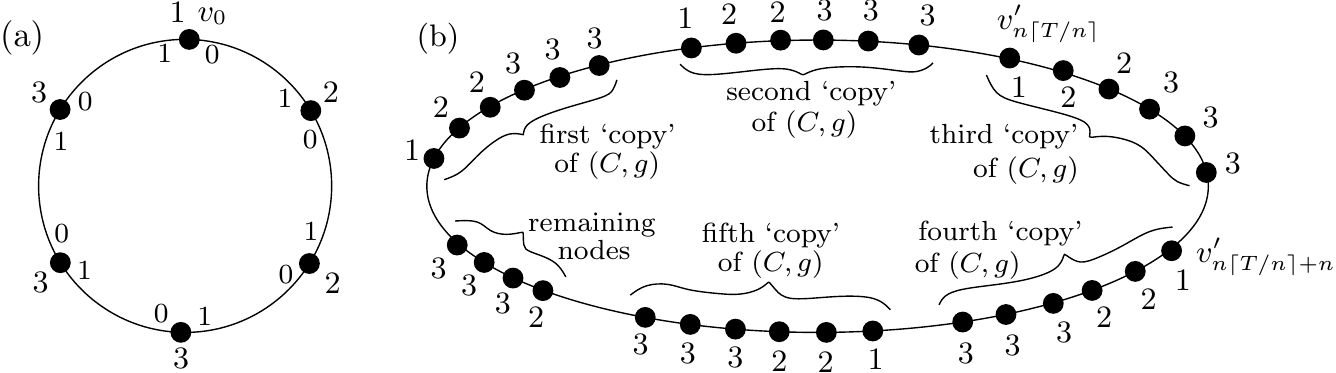}
\caption{(a) an example of a ring $(C,g)$ obtained from $(x_1,x_2,x_3)=(1,2,3)$, $n=6$; (b) the corresponding ring $(C',g')$ constructed for $T\in\{7,\ldots,12\}$ and $(x_1',x_2',x_3')=(5,11,18)$, $n'=34$.}
\label{fig:ring}
\end{center}
\end{figure}

Let $v_0$ be any node of $C$.
Let $T$ be the time after which $A$ produces the answer.
Since the restriction set $R$ is unbounded, there exists $(x_1',\ldots,x_r')\in R$ such that $x_i'\geq 2x_i\lceil T/n\rceil+x_i$ for each $i\in\{1,\ldots,r\}$.
We construct a colored ring $(C',g')$ on $n'=x_1'+\cdots+x_r'$ nodes $v_0',\ldots,v_{n'-1}'$ such that:
\begin{itemize}
 \item $v_i'$ is adjacent to $v_{(i+1)\modulo n'}'$, $i\in\{0,\ldots,n'-1\}$,
 \item the port numbers of the edge $\{v_i',v_{i+1}'\}$ at $v_i'$ and $v_{i+1}'$ are equal to the port numbers of the edge $\{v_{i\modulo n},v_{(i+1)\modulo n}\}$ at $v_{i\modulo n}$ and $v_{(i+1)\modulo n}$, respectively, $i\in\{0,\ldots,2n\lceil T/n\rceil+n-1\}$, in $C$,
 \item the port numbers of the remaining edges are set arbitrarily in a way that guarantees proper port labeling.
\end{itemize}
Moreover, the colors are assigned to the nodes of $C'$ as follows:
\[g'(v_{i}')=g(v_{i\modulo n}), \quad i\in\{0,\ldots,2n \lceil T/n\rceil+n-1\},\]
and the remaining nodes, i.e., the ones in $X=\{v_{2n \lceil T/n\rceil+n}',\ldots,v_{n'}'\}$, receive colors in any way that ensures that the size of color $j\in\{1,\ldots,r\}$ in $C'$ is $x_j'$.
See Figure~\ref{fig:ring} for an example of the construction of $(C',g')$.
The colored view of depth $T$ of $v_0$ in $(C,g)$ is the same as the colored views of depth $T$ of $v_{n \lceil T/n\rceil}'$ and of $v_{n \lceil T/n\rceil+n}'$ in $(C',g')$ because $T\leq n \lceil T/n\rceil$.

Consider the problem $\problemTOP$.
If the answer produced by $A$ in $(C,g)$ is ``unsolvable'', then the lemma follows.
Hence, we may assume that $A$ returns the topology of $(C,g)$.
By Proposition \ref{prop:algoView}, the algorithm $A$ executed by $v_{n \lceil T/n\rceil}'$ in $(C',g')$ stops after time $T$ and produces the same answer as for $v_0$ in $(C,g)$.
Thus, $A$ must also return the topology of $(C,g)$ when executed by $v_{n\lceil T/n\rceil}'$ in $(C',g')$.
Since $C$ and $C'$ are of different sizes, we obtain a contradiction, as required.

Next consider the problem $\problemLE$.
If the answer produced by $A$ in $(C,g)$ is ``unsolvable'', then the lemma follows.
Hence, we may assume that $A$ elects a leader in $(C,g)$.
By Proposition \ref{prop:algoView}, the algorithm $A$ executed by $v_{n \lceil T/n\rceil}'$ and $v_{n \lceil T/n\rceil+n}'$ in $(C',g')$ stops after time $T$ and produces the same answer as in $(C,g)$.
Thus, these two nodes elect different leaders --- a contradiction.
\end{proof}

We next turn attention to the issue of time needed to solve problems $\problemTOP$ and $\problemLE$, assuming that an upper bound $k$ on the size of some color is given to all nodes.  We give a lower bound showing that the time
$O(kD+D\log(n/D))$ of our algorithm is optimal.
We first construct a class of networks, for which
time $\Theta(kD)$ cannot be improved.

\begin{proposition}\label{first-lb}
Let $k,D \leq n$ be arbitrary positive integers.
There exists a network of size $\Theta(n)$ and diameter $\Theta(D)$, whose nodes, if they are given as input an upper bound $k$ on the size of one color $\ourColor$ and have no other information on the network, need time  $\Omega(kD)$ to solve the problems $\problemTOP$ and $\problemLE$.
\end{proposition}
\begin{proof}
For any $n'$ and $d<n'/2$, we define a $n'$-node graph $G(n',d)$ (called a chordal ring).
Denote by $v_{0},\ldots,v_{n'-1}$ the nodes of $G(n',d)$.
For each $i\in\{0,\ldots,n'-1\}$ and $j\in\{1,\ldots,d\}$, let $\{v_{i},v_{i+j}\}$ be an edge of $G(n',d)$, where the port number of this edge at $v_i$ is $j-1$ and the port number at $v_{i+j}$ is $d+j-1$.
Note that the diameter of $G(n',d)$ is $\Theta(n'/d)$.

Since our result {holds asymptotically}, we may assume that $D\geq 3$ and $k\geq 12$.
Consider an algorithm $A$ for solving problem $\problemTOP$ or $\problemLE$ in any network.
Let $d=\lfloor n/D\rfloor$.
(Note that $D\geq 3$ implies $d<n/2$, as required in the construction of $G(n,d)$.)
Let $u$ be any node of $G(n,d)$.
Let $(G(n,d),g)$ be the colored network in which $g(u)=\ourColor$, and all other nodes of $G(n,d)$ have the same color $\ourColor'$, different from $\ourColor$.
Note that the diameter of $G(n,d)$ is $\Theta(D)$.
We run the algorithm $A$ in the colored network $(G(n,d),g)$.
We argue that $A$ should run for time $\Omega(kD)$.
Suppose for a contradiction that $A$ stops and produces an answer in round $T<\lfloor k/3\rfloor D$.

Consider the network $G(kn,d)$ on the set of nodes $\{v_0',\ldots,v_{kn-1}'\}$.
Construct a network $G'$ by adding a pendant edge to the node $v_0'$ of $G(kn,d)$, i.e., add an extra node of degree $1$ and attach it to $v_0'$.
Let $u'=v_{\lfloor k/2\rfloor n}'$.
Let $(G',g')$ be a colored network in which $g'(v_{nj}')=\ourColor$ for each $j\in\{0,\ldots,k-1\}$, while all other nodes of this network have color $\ourColor'$.
By construction of $G'$, the distance between $u'$ and $v_0'$ in $G'$ is greater than $T$.
Thus, the truncated colored view $\view[T](u')$ in $(G',g')$ is the same as the truncated colored view $\view[T](u)$ in $(G(n,d),g)$.
Now, we run $A$ in the colored network $(G',g')$.

Suppose that $A$ is an algorithm solving the problem $\problemTOP$.
By Proposition \ref{prop:algoView}, $A$ executed on $u$ and $u'$ produces the same answer to problem $\problemTOP$ in $G(n,d)$ and $G'$, respectively.
If the answer on $u$ is the topology of $G(n,d)$, then we immediately have a contradiction since the two networks are not isomorphic.
On the other hand, if the answer is ``unsolvable'', then this answer is incorrect for $G'$.
The latter is due to Corollary \ref{cor:main}.

Let now $A$ be an algorithm solving the problem $\problemLE$.
Let $v'=v_{\lfloor kn/2\rfloor+n}'$.
The distance from $v'$ to $v_0'$ in $G'$ is at least $(\lfloor k/2\rfloor-1)\lfloor n/d\rfloor\geq \lfloor n/d\rfloor\cdot k/3\geq Dk/3\geq T$ for $k\geq 12$.
Thus, $\view[T](u')=\view[T](v')$.
By Proposition \ref{prop:algoView}, $A$ produces the same output at $u'$ and $v'$ after time $T$.
By Corollary \ref{cor:main}, the algorithm $A$ cannot output ``unsolvable'' because all nodes of $G'$ have unique views and the size of color $\ourColor$ is $k$ in $(G',g')$.
Thus, $u'$ and $v'$ elect different leaders --- a contradiction.
\end{proof}

The other part of our lower bound follows from \cite{DKP}.
\begin{proposition}\label{second-lb}
Let $D \leq n$ be arbitrary positive integers. There exists a network of size $\Theta(n)$
 and diameter $\Theta(D)$, whose nodes need time at least $\Omega(D\log(n/D))$ to solve problems $\problemTOP$ and $\problemLE$, even if all nodes have the same color and they are given the size and the diameter of the 
 network.
\end{proposition}
\begin{proof}
By Proposition \ref{trunc-lower}, there exists a network $G$ of size $\Theta(n)$ and diameter $\Theta(D)$, with nodes $u$ and $v$, both with unique views,
such that $\cV(u)\neq \cV(v)$ but $\cV^{h'}(u)=\cV^{h'}(v)$, for some $h' \in \Theta(D \log (n/D))$.
Thus, by Proposition \ref{prop:algoView}, any algorithm $A$ that stops after at most $h'$ steps and produces an answer to problem $\problemLE$, gives the same answer at $u$ and $v$.
Thus, this answer must be ``unsolvable''.
(Otherwise, two distinct leaders would be elected.)
However, since the nodes of $G$ have pairwise different views and the size of the network is known, Corollary~\ref{cor:main} implies that leader election is possible in this network.
Thus, any algorithm solving problem $\problemLE$ needs time $\Omega(D\log(n/D))$ in $G$.

Now consider the problem $\problemTOP$ and let $A$ be any algorithm solving this problem.
Suppose for a contradiction that $A$ stops after at most $h'$ steps. 
By Proposition~\ref{trunc-lower}, there exists a network $G'$, different than $G$, with the same size and diameter as $G$, with a node $u'$, such that $\view[h'](u')=\view[h'](u)$.
Thus, by Proposition~\ref{prop:algoView}, $A$ returns the same answer at $u$ and $u'$.
Since $G$ and $G'$ are different, this answer must be ``unsolvable''.
Since all nodes in $G$ have pairwise different views, by Proposition~\ref{prop:algoView} (where $k$ is {taken} to be the size of $G$), $\problemTOP$ is possible in $G$ --- a contradiction.
\end{proof}

Theorem \ref{thm:main}, together with Propositions \ref{first-lb} and \ref{second-lb}, imply the following corollary showing that our algorithm is time-optimal.

\begin{corollary}
The optimal time to solve problems $\problemTOP$ and $\problemLE$ on $n$-node networks with diameter $D$,
assuming that nodes know only an upper bound $k$ on the size of a given color, is $\Theta(kD+D\log(n/D))$.  
\end{corollary} 

\section{Conclusion}

We showed that nodes of a colored network can solve problems $\problemTOP$ and $\problemLE$, if they are given an upper bound on the number of nodes of a given color, and we studied the time
of solving these problems in the $\cal{LOCAL}$  model, under this assumption. 

Notice that the synchronous behavior of the $\cal{LOCAL}$  model can be easily reproduced in an asynchronous network, by defining, for each node $u$ separately, an asynchronous round $i$ consisting of the following actions of this node: node $u$ performs local computations, then sends messages stamped with integer $i$ to all its neighbors, and waits for messages stamped $i$ from all neighbors.
In order to implement this, every node must send a message with all consecutive stamps, until termination, some of the messages possibly empty.
Our results concerning time of solving problems $\problemTOP$ and $\problemLE$ can be translated for asynchronous networks by replacing ``the number of rounds''  by ``the maximum number of asynchronous rounds, over all nodes''.

Let $D$ be the diameter of the network.
If nodes have distinct labels, then time $D+1$ in the $\cal{LOCAL}$  model is enough to solve any problem solvable on a given network, as after this time all nodes solve topology recognition.
By contrast, in our scenario of colored nodes, time $D+1$ is often not enough, for example to elect a leader, or to perform topology recognition, even if these tasks are feasible.
This is due to the fact that after time $t\leq D+1$ each node may learn only all colored paths of length $t$ originating at it.
Acquiring this information does not imply getting a picture of the radius $t$ colored neighborhood of the node.
This is because a node $v$ may not know if two paths originating at it have the same other endpoint or not.
We showed that these ambiguities may force time much larger than $D$ to solve problems $\problemTOP$ and $\problemLE$.

As it is always assumed in the $\cal{LOCAL}$ model, we allowed arbitrarily large messages to be sent in each round. Bounding the size of messages to logarithmic in the size of the network, as
it is assumed in the alternative $\cal{CONGEST}$ model, would likely have an important impact on the time of solving problems $\problemTOP$ and $\problemLE$. Hence an interesting open question is
to establish the best time of solving these problems in the latter model.

\bibliographystyle{plain}

\begin{thebibliography}{99}

\bibitem{An}
D.~Angluin, Local and Global Properties in Networks of Processors,
{\em Proc. 12th Annual ACM Symposium on Theory of Computing} (STOC 1980), 82--93.

\bibitem{ABDKPR}
H. Attiya, A. Bar-Noy, D. Dolev, D. Koller, D. Peleg, R. Reischuk,
Renaming in an Asynchronous Environment, {\em Journal of the ACM} 37 (1990), 524--548.

\bibitem{ASW}
H. Attiya, M. Snir, M. Warmuth,
Computing on an Anonymous Ring,
{\em Journal of the ACM} 35 (1988), 845-875.

\bibitem{AtSn}
H. Attiya, M. Snir,
Better Computing on the Anonymous Ring,
{\em Journal of Algorithms} 12 (1991), 204-238.

\bibitem{A}
B. Awerbuch, Optimal Distributed Algorithms for Minimum Weight Spanning Tree, 
Counting, Leader Election and Related Problems,
{\em Proc. 19th Annual ACM Symposium on Theory of Computing} (STOC 1987), 230-240.

\bibitem{AGPV}
B. Awerbuch, O. Goldreich, D. Peleg, R. Vainish,
A Trade-Off between Information and Communication in Broadcast Protocols,
{\em J. {ACM}} 37 (1990), 238-256.

\bibitem{BSVCGS}
P. Boldi, S. Shammah, S. Vigna, B. Codenotti, P. Gemmell, J. Simon,
Symmetry Breaking in Anonymous Networks: Characterizations,
{\em Proc. 4th Israel Symposium on Theory of Computing and Systems} (ISTCS 1996), 16-26.

\bibitem{BV}
P. Boldi, S. Vigna,
Computing Anonymously with Arbitrary Knowledge,
{\em Proc. 18th ACM Symp. on Principles of Distributed Computing} (PODC 1999), 181-188.

\bibitem{B}
J.E. Burns, A Formal Model for Message Passing Systems,
{\em Tech. Report TR-91}, Computer Science Department,
Indiana University, Bloomington, September 1980.

\bibitem{C}
J. Chalopin,
Local Computations on Closed Unlabelled Edges: The Election Problem and the Naming Problem,
{\em Proc. 31st Conference on Current Trends in Theory and Practice of Computer Science} (SOFSEM 2005), 82-91.

\bibitem{CDK}
J. Chalopin, S. Das, A. Kosowski, 
Constructing a Map of an Anonymous Graph: Applications of Universal Sequences,
{\em Proc. 14th International Conference on Principles of Distributed Systems} (OPODIS 2010), 119-134.

\bibitem{CMM}
J. Chalopin, A.W. Mazurkiewicz, Y. M\'etivier, Labelled (Hyper)Graphs, Negotiations and the Naming Problem,
{\em Proc. 4th International Conference on Graph Transformations} (ICGT 2008), 54-68.

\bibitem{CM}
J. Chalopin, Y. M\'etivier,
Election and Local Computations on Edges,
{\em Proc. Foundations of Software Science and Computation Structures} (FoSSaCS 2004), 90-104.

\bibitem{DKP}
D. Dereniowski, A. Kosowski, D. Pajak, Distinguishing Views in Symmetric Networks: A Tight Lower Bound,
{\em Theoretical Computer Science} 582 (2015) 27-34.

\bibitem{DP}
D. Dereniowski, A. Pelc, Drawing Maps with Advice,  {\em Journal of Parallel and Distributed Computing} 72 (2012), 132-143. 

\bibitem{DP1}
D. Dereniowski, A. Pelc, Leader Election for Anonymous Asynchronous Agents in Arbitrary Networks, {\em Distributed Computing} 27 (2014), 21-38. 

\bibitem{DKMP}
K. Diks, E. Kranakis A. Malinowski, A. Pelc,
Anonymous Wireless Rings,
{\em Theoretical Computer Science} 145 (1995), 95-109.

\bibitem{DoPe}
S. Dobrev, A. Pelc, 
Leader Election in Rings with Nonunique Labels, {\em Fundamenta Informaticae} 59 (2004), 333-347. 

\bibitem{FKKLS}
P. Flocchini, E. Kranakis, D. Krizanc, F.L. Luccio, N. Santoro,
Sorting and Election in Anonymous Asynchronous Rings,
{\em Journal of Parallel and Distributed Computing} 64 (2004), 254-265.

\bibitem{FL}
G.N. Fredrickson, N.A. Lynch,
Electing a Leader in a Synchronous Ring,
{\em Journal of the ACM} 34 (1987), 98-115.

\bibitem{FP}
E. Fusco, A. Pelc, How Much Memory is Needed for Leader Election, {\em Distributed Computing} 24 (2011), 65-78. 

\bibitem{FP1}
E. Fusco, A. Pelc, Knowledge, Level of Symmetry, and Time of Leader Election,
{\em Proc. 20th Annual European Symposium on Algorithms} (ESA 2012), LNCS 7501, 479-490. 

\bibitem{FPR}
E. Fusco, A. Pelc, R. Petreschi, Use Knowledge to Learn Faster: Topology Recognition with Advice,
{\em Proc. 27th International Symposium on Distributed Computing} (DISC 2013), LNCS 8205, 31-45. 

\bibitem{GMP}
C. Glacet, A. Miller, A. Pelc, Time vs. information tradeoffs for leader election in anonymous trees, Proc. 27th Annual ACM-SIAM Symposium on Discrete Algorithms (SODA 2016), 600-609.




\bibitem{HKMMJ}
M.A. Haddar, A.H. Kacem, Y. M\'{e}tivier, M. Mosbah, M. Jmaiel, Electing a Leader in the Local Computation Model using Mobile Agents,
{\em Proc.  6th ACS/IEEE International Conference on Computer Systems and Applications} (AICCSA 2008), 473-480.

\bibitem{H}
J. Hendrickx, Views in a Graph: To Which Depth Must Equality Be Checked?,
{\em IEEE Transactions on Parallel and Distributed Systems} 25 (2014) 1907-1912.


\bibitem{HS}
D.S. Hirschberg, J.B. Sinclair,
Decentralized Extrema-Finding in Circular Configurations of Processes,
{\em Communications of the ACM} 23 (1980), 627-628.

\bibitem{JKZ}
T. Jurdzinski, M. Kutylowski, J. Zatopianski, 
Efficient Algorithms for Leader Election in~Radio Networks,
 {\em Proc., 21st ACM Symp. on Principles of Distributed Computing} (PODC 2002), 51-57.

\bibitem{KP}
D. Kowalski, A. Pelc, Leader Election in Ad Hoc Radio Networks: A Keen Ear Helps, 
{\em Proc. 36th International Colloquium on Automata, Languages and Programming} (ICALP 2009), LNCS 5556, 521-533. 


\bibitem{Kr}
E. Kranakis,
Symmetry and Computability in Anonymous Networks: A Brief Survey,
{\em Proc. 3rd Int. Conf. on Structural Information and Communication Complexity}, (1997), 1-16.

\bibitem{KKV}
E. Kranakis, D. Krizanc, J. van der Berg,
Computing Boolean Functions on Anonymous Networks,
{\em Information and Computation} 114 (1994), 214-236.

\bibitem{LL}
G. Le Lann,
Distributed Systems - Towards a Formal Approach,
{\em Proc. IFIP Congress}, North Holland, (1977), 155-160.

\bibitem{Ly}
N.A. Lynch, 
Distributed Algorithms,
Morgan Kaufmann Publ., Inc., 1996.


\bibitem{MP}
A. Miller, A. Pelc: Election vs. selection: Two ways of finding the largest node in a graph,
CoRR abs/1411.1319 (2014).

\bibitem{NO}
K. Nakano, S. Olariu, Uniform Leader Election Protocols for Radio Networks,
{\em IEEE Transactions on Parallel and Distributed Systems} 13
(2002), 516-526.

\bibitem{Norris}
N. Norris, Universal Covers of Graphs: Isomorphism to Depth $N-1$ Implies Isomorphism to All Depths,
{\em Discrete Applied Mathematics} 56 (1995), 61-74.
               
\bibitem{Pe}D. Peleg,
Distributed Computing, A Locality-Sensitive Approach,
SIAM Monographs on Discrete Mathematics and Applications, Philadelphia 2000.

\bibitem{P}
G.L. Peterson, An $O(n \log n)$ Unidirectional Distributed Algorithm for the Circular Extrema Problem,
{\em ACM Transactions on Programming Languages and Systems} 4 (1982), 758-762.

\bibitem{Saka}
N. Sakamoto,
Comparison of Initial Conditions for Distributed Algorithms
on Anonymous  Networks,
{\em Proc. 18th ACM Symp. on Principles of Distributed Computing} (PODC 1999), 173-179.

\bibitem{Wil}
D.E. Willard, 
Log-logarithmic Selection Resolution Protocols in a Multiple Access Channel,
{\em SIAM J. on Computing} 15 (1986), 468-477. 

\bibitem{YK}
M. Yamashita, T. Kameda,
Computing on Anonymous Networks,
{\em Proc. 7th ACM Symp. on Principles of Distributed Computing} (PODC 1988), 117-130.

\bibitem{YK2}
M. Yamashita, T. Kameda,
Electing a Leader when Procesor Identity Numbers are not Distinct,
{\em Proc. 3rd Workshop on Distributed Algorithms} (WDAG 1989), LNCS 392, 303-314.

\bibitem{YK3}
M. Yamashita and T. Kameda,
Computing on Anonymous Networks: Part I - Characterizing the Solvable Cases,
{\em IEEE Trans. Parallel and Distributed Systems} 7 (1996), 69-89. 

\end{thebibliography}

\end{document}